\documentclass[11pt]{article}
\usepackage{authblk}
\usepackage{mathrsfs}
\usepackage{hyperref}
\usepackage{enumerate}
\usepackage{algorithm}
\usepackage{algpseudocode}
\usepackage{epsfig}
\usepackage{wrapfig}
\usepackage{amsmath}
\usepackage{amssymb}
\usepackage{latexsym}
\usepackage{amsfonts}
\usepackage{amsthm}
\usepackage[numbers,sort&compress]{natbib}
\usepackage{graphicx}
\usepackage{color}
\usepackage{fullpage}

\newif\ifFull 
\Fulltrue

\newcommand {\opt } {\mathrm{OPT}}

\def\argmin{\mathop{\rm argmin}}

\newtheorem{definition}{Definition}

\newtheorem{lemma}{Lemma}

\newtheorem{theorem}[lemma]{Theorem}

\newtheorem{observation}{Observation}

\newtheorem{claim}[lemma]{Claim}
\newenvironment{remark}[1][Remark]{\begin{trivlist}
\item[\hskip \labelsep {\bfseries #1}]}{\end{trivlist}}

\begin{document}

\title{
{Linear-time approximation schemes for planar minimum three-edge connected and three-vertex connected spanning subgraphs }}
\author{Baigong Zheng }
\affil{Oregon State University \\
zhengb@eecs.oregonstate.edu
}
\date{}
\maketitle
\thispagestyle{empty}
\begin{abstract}
We present the first polynomial-time approximation schemes, i.e., $(1 + \epsilon)$-approximation algorithm for any constant $\epsilon > 0$, for the minimum three-edge connected spanning subgraph problem and the minimum three-vertex connected spanning subgraph problem in undirected planar graphs.
Both the approximation schemes run in linear time.

\end{abstract}
\vfill
{\small 
This material is based upon work supported by the National Science Foundation under Grant No.\ CCF-1252833.}

\newpage
\setcounter{page}{1}

\section{Introduction}

Given an undirected unweighted graph $G$, the {\em minimum $k$-edge connected spanning subgraph problem} ($k$-ECSS) asks for a spanning subgraph of $G$ that is $k$-edge connected (remains connected after removing any $k-1$ edges) and has a minimum number of edges. 
The {\em minimum $k$-vertex connected spanning subgraph problem} ($k$-VCSS) asks for a $k$-vertex connected (remains connected after removing any $k-1$ vertices) spanning subgraph of $G$ with minimum number of edges.
These are fundamental problems in network design and have been well studied.
When $k=1$, the solution is simply a spanning tree for both problems.
For $k \ge 2$, the two problems both become NP-hard~\cite{GJ79,CT00},
so people put much effort into achieving polynomial-time approximation algorithms.
Cheriyan and Thurimella~\cite{CT00} give algorithms with approximation ratios of $1+1/k$ for $k$-VCSS and $1+2/(k+1)$ for $k$-ECSS for simple graphs.
Gabow and Gallagher~\cite{GG12} improve the approximation ratio for $k$-ECSS to $1 + 1/(2k) + O(1/k^2)$ for simple graphs when $k\ge 7$, and they give a $(1 +21/(11k))$-approximation algorithm for $k$-ECSS in multigraphs.
Some researchers have studied these two problems for the small connectivities $k$, especially $k=2$ and $k=3$, and obtained better approximations.  
The best approximation ratio for $2$-ECSS in general graphs is $5/4$ of Jothi, Raghavachari, and Varadarajan~\cite{JRV03}, while for $2$-VCSS in general graphs, the best ratio is $9/7$ of Gubbala and Raghavachari~\cite{GR05}.
Gubbala and Raghavachari~\cite{GR07} also give a $4/3$-approximation algorithm for 3-ECSS in general graphs.

A {\em polynomial-time approximation scheme} (PTAS) is an algorithm that, given an instance and a positive number $\epsilon$, finds a $(1+\epsilon)$-approximation for the problem and runs in polynomial time for fixed $\epsilon$.
Neither $k$-ECSS nor $k$-VCSS have a PTAS even in graphs of bounded degree for $k =2$ unless P = NP~\cite{CL99}.
However, this hardness of approximation does not hold for special classes of graphs and small values of $k$.
For example, Czumaj et al.~\cite{CGSZ04} show that there are PTASes for both of 2-ECSS and 2-VCSS in planar graphs. Both problems are NP-hard in planar graphs (by a reduction from Hamiltonian cycle).
Later, Berger and Grigni improved the PTAS for $2$-ECSS to run in linear time~\cite{BG07}.


Following their PTASes for 2-ECSS and 2-VCSS, Czumaj et al.~\cite{CGSZ04} ask the following: can we extend the PTAS for $2$-ECSS to a PTAS for 3-ECSS in planar graphs?
A PTAS for 3-VCSS in planar graphs is additionally listed as an open problem in the {\em Handbook of Approximation Algorithms and Metaheuristics} (Section 51.8.1)~\cite{Gonzalez2007}.
In this paper we answer these questions affirmatively by giving the first PTASes for both 3-ECSS and 3-VCSS in planar graphs.
Our main results are the following theorems.
\begin{theorem}\label{thm: EC}
For 3-ECSS, there is an algorithm that, for any $\epsilon > 0$ and any undirected planar graph $G$, finds a $(1 + \epsilon)$-approximate solution in linear time.
\end{theorem}
\begin{theorem}\label{thm: VC}
For 3-VCSS, there is an algorithm that, for any $\epsilon > 0$ and any undirected planar graph $G$, finds a $(1 + \epsilon)$-approximate solution in linear time.
\end{theorem}
In the following, we assume there are no self-loops in the input graph for both of 3-ECSS and 3-VCSS.
For $3$-ECSS, we allow parallel edges in $G$, but at most $3$ parallel edges between any pair of vertices are useful in a minimal solution.
For $3$-VCSS, parallel edges are unnecessary, so we assume the input graph is simple. 
Since three-vertex connectivity (triconnectivity) and three-edge connectivity can be verified in linear time~\cite{Schmidt13, MNS13}, we assume the input graph $G$ contains a feasible solution.
W.l.o.g. we also assume $\epsilon <1$.

\subsection{The approach}
Our PTASes follow the general framework for planar PTASes that grew out of the PTAS for TSP~\cite{Klein08}, and has been applied to obtain PTASes for other problems in planar graphs, including minimum-weight 2-edge-connected subgraph~\cite{BG07}, Steiner tree~\cite{BKK07, BKM09}, Steiner forest~\cite{BHM11} and relaxed minimum-weight subset two-edge connected subgraph~\cite{BK08}.
The framework consists of the following four steps.
\begin{description}
\item [Spanner Step.] Find a subgraph $G'$ that contains a $(1+\epsilon)$-approximation and whose total weight is bounded by a constant times of the weight of an optimal solution. Such a graph is usually called a {\em spanner} since it often approximates the distance between vertices.
\item [Slicing Step.] Find a set of subgraphs, called {\em slices}, in $G'$ such that any two of them are face disjoint and share only a set of edges with small weight and each of the subgraphs has bounded {\em branchwidth}. 
\item [Dynamic-Programming Step.] 
Find the optimal solution in each slice using dynamic programming.
Since the branchwidth of each subgraph is bounded, the dynamic programming runs in polynomial time.
\item [Combining Step.] Combine the solutions of all subgraphs obtained by dynamic programming and some shared edges from Slicing step to give the final approximate solution.
\end{description}

For most applications of the PTAS framework, the challenging step is to illustrate the existence of a spanner subgraph.
However, for 3-ECSS and 3-VCSS, we could simply obtain a spanner from the input graph $G$ by deleting additional parallel edges since by planarity there are at most $O(n)$ edges in $G$ and the size of an optimal solution is at least $n$, where $n = |V(G)|$.
So, different from those previous applications, the real challenge for our problems is to illustrate the slicing step and the combining step.  
For the slicing step, we want to identify a set of slices that have two properties: (1) three-edge-connectivity for 3-ECSS or triconnectivity for 3-VCSS, and (2) bounded branchwidth.  
With these two properties, we can solve 3-ECSS or 3-VCSS on each slice efficiently.
For the combining step, we need to show that we can obtain a nearly optimal solution from the optimal solutions of all slices found in slicing step and the shared edges of slices,
that means the solution should satisfy the connectivity requirement and its size is at most $1+\epsilon$ times of the size of an optimal solution for the original input graph.

To identify slices, we generalize a decomposition used by Baker~\cite{Baker94}.
Before sketching our method, we briefly mention the difficulty in applying previous techniques.
The PTAS for TSP~\cite{Klein08} identifies slices in a spanner $G'$ by doing a breadth-first search in its planar dual to decompose the edge set into some levels, and 
any two adjacent slices could only share all edges of the same level, which form a set of edge-disjoint simple cycles.
This is enough to achieve simple connectivity or biconnectivity between vertices of different slices. 
But our problems need stronger connectivity for which one cycle is not enough. 
For example, we may need a non-trivial subgraph outside of slice $H$ to maintain the triconnectivity between two vertices in slice $H$. 
See Figure~\ref{fig: previous} (a).
\begin{figure}
\centering
\includegraphics[scale = 1]{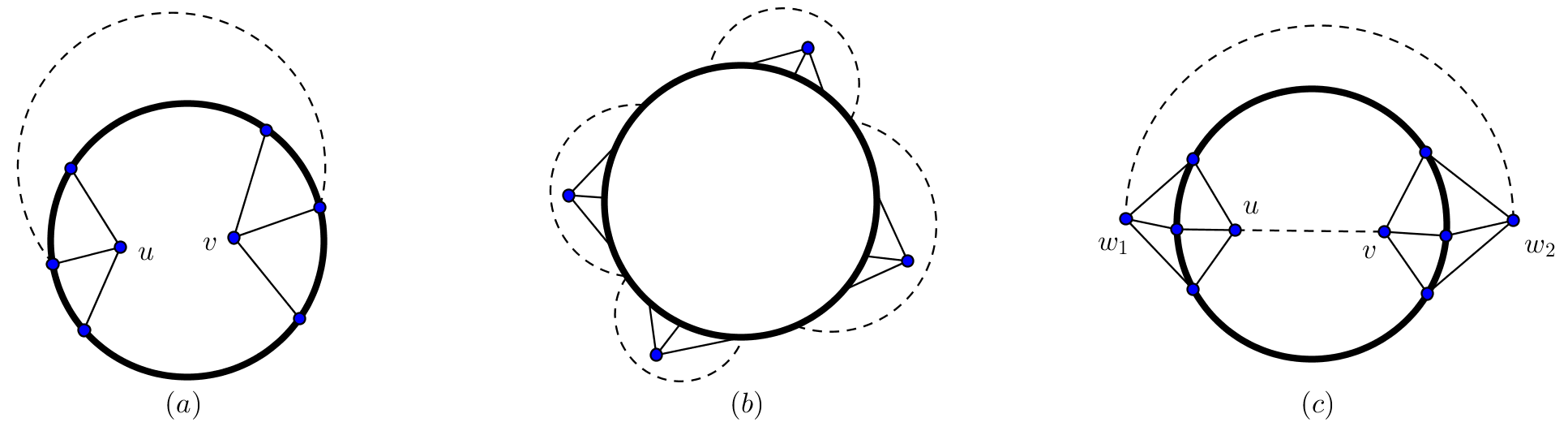}
\caption{(a) The bold cycle encloses a slice $H$. To maintain the triconnectivity between two vertices $u$ and $v$ in $H$, we need the dashed path outside of $H$. 
(b) The bold cycle encloses a 3EC slice $H$. The dashed edges divide the outer face of $H$ into distinct regions, which may contain contracted nodes.
(c) The bold cycle encloses a 3EC slice $H$. The dashed path $P$ between $w_1$ and $w_2$ will be contracted to obtain a node $x$. A solution for 3-ECSS on $H$ may contain $x$ but not the dashed edge between $u$ and $v$. Then the union $S$ of feasible solutions on all 3EC slices is not feasible if it does not contain path $P$. 
}
\label{fig: previous}
\end{figure}

For 3-ECSS, we construct a graph called {\em 3EC slice}.
We contract each component of the spanner that is not in the current 3EC slice.
Since contraction can only increase edge-connectivity, this will give us the 3EC slices that are three-edge connected (Lemma~\ref{lem: EC1} in Section~\ref{sec: EC}). 
However, if we directly apply this contraction method in the slicing step of the PTAS for TSP, the branchwidth of the 3EC slice may not be bounded.
This is because there may be many edges that are not in the slice but have both endpoints in the slice, and such edges may divide the faces of the slice into distinct regions, each of which may contain a contracted node.
See Figure~\ref{fig: previous} (b).
To avoid this problem, we apply the contractions in the decomposition used by Baker~\cite{Baker94}, which define a slice based on the levels of vertices instead of edges.
We can prove that each 3EC slice has bounded branchwidth in this decomposition (Lemma~\ref{lem: bw} in Section~\ref{sec: slice}).

Although each 3EC slice is three-edge connected, the union $S$ of their feasible solutions may not be three-edge connected.
Consider the following situation.
In a solution for a slice, a contracted node $x$ is contained in all vertex-cuts for a pair of vertices $u$ and $v$.
But in $S$, the subgraph induced by the vertex set $X$ corresponding to $x$ may not be connected.
Therefore, $u$ and $v$ may not satisfy the connectivity requirement in $S$.
See Figure~\ref{fig: previous} (c).

In their paper, Czumaj et al.~\cite{CGSZ04} proposed a structure called {\em bicycle}. A bicycle consists of two nested cycles, and all in-between faces visible from one of the two cycles. This can be used to maintain the three-edge-connectivity between those connecting endpoints in cycles.
This motivates our idea: we want to combine this structure with Baker's shifting technique so that two adjacent slices shared a subgraph similar to a bicycle. In this way, we could maintain the strong connectivity between adjacent slices by including all edges in this shared subgraph, whose size could be bounded by the shifting technique.
Specifically, we define a structure called {\em double layer} for each level $i$ based on our decomposition, which intuitively contains all the edges incident to vertices of level $i$ and all edges between vertices of level $i+1$.
Then we define a 3EC slice based on a maximal circuit 
such that any pair of 3EC slices can only share edges in the double layer between them.
In this way, we can obtain a feasible solution for 3-ECSS by combining the optimal solutions for all the slices and all the shared double layers (Lemma~\ref{lem: EC2} in Section~\ref{sec: EC}). 
By applying shifting technique on double layers, we can prove that the total size of the shared double layers is a small fraction of the size of an optimal solution.
So we could add those shared double layers into our solution without increasing its size by much, and this gives us a nearly optimal solution.

For 3-VCSS, we construct a {\em 3VC slice} based on a simple cycle instead of a circuit.
The construction is similar to that of 3EC slices. However, contraction does not maintain vertex connectivity. So we need to prove each 3VC slice is triconnected~(Lemma~\ref{lem: VC1} in Section~\ref{sec: VC}).
Then similar to 3-ECSS, we also need to prove that the union of the optimal solutions of all 3VC slices and all shared double layers form a feasible solution (Lemma~\ref{lem: VC2} in Section~\ref{sec: VC}). 

For dynamic-programming step, we need to solve a minimum-weight 3-ECSS problem in each 3EC slice and a minimum-weight 3-VCSS problem in each 3VC slice. 
This is because we need to carefully assign weights to edges in a slice so that we can bound the size of our solution.
We provide a dynamic program for the minimum-weight 3-ECSS problem in graphs of bounded branchwidth in \ifFull Section~\ref{sec: dp}, \fi which is similar to that in the works of Czumaj and Lingas~\cite{CL98, CL99}.
A dynamic program for minimum-weight 3-VCSS can be obtained in a similar way.
Then we have the following theorem.
\begin{theorem}\label{thm: dp}
Minimum-weight 3-ECSS problem and minimum-weight 3-VCSS problem both can be solved on a graph $G$ of bounded branchwidth in $O(|E(G)|)$ time.
\end{theorem}


\section{Preliminaries}\label{sec: slice}
Let $G$ be an undirected planar graph with vertex set $V(G)$ and edge set $E(G)$.
We denote by $G[S]$ the subgraph of $G$ induced by $S$ where $S$ is a vertex subset or an edge subset.
We simplify $|E(G)|$ to $|G|$.
We assume we are given an embedding of $G$ in the plane.
We denote by $\partial (G)$ the subgraph induced by the edges on the outer boundary of $G$ in this embedding. 
A {\em circuit} is a closed walk that may contain repeated vertices but not repeated edges.
A {\em simple cycle} is a circuit that contains no repetition of vertices, other than the repetition of the starting and ending vertex.
A simple cycle bounds a finite region in the plane that is a topological disk.
We say a simple cycle {\em encloses} a vertex, an edge or a subgraph if the vertex, edge or subgraph is embedded in the topological disk bounded by the cycle.
We say a circuit {\em encloses} a vertex, edge or subgraph if the vertex, edge or subgraph is enclosed by a simple cycle in the circuit.

\begin{figure}
\centering
\includegraphics[scale = 1]{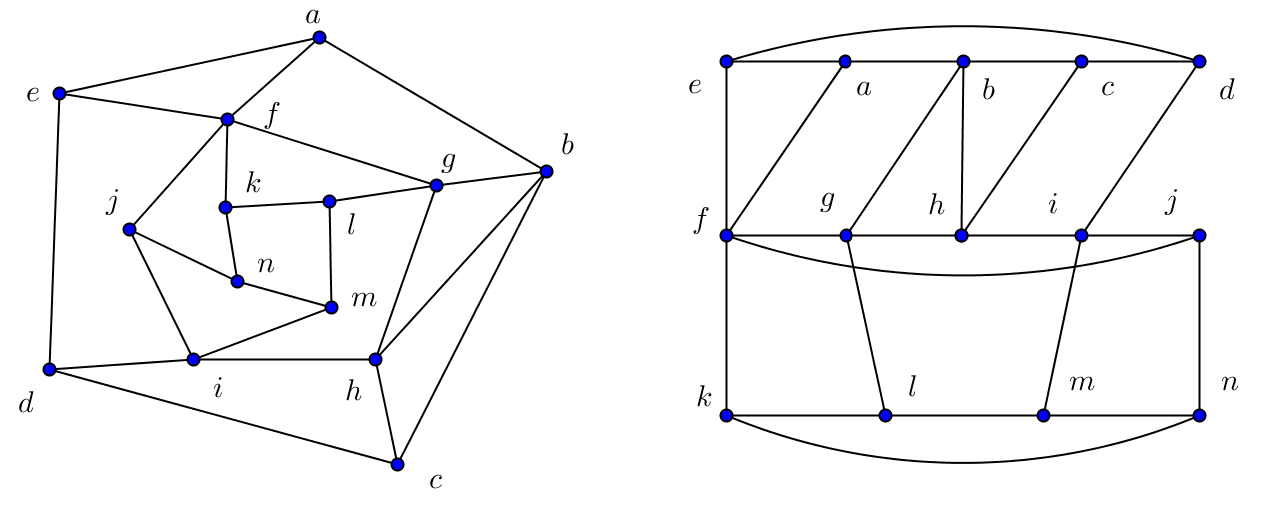}
\caption{The three horizontal lines in the right figure show the three levels in the left figure. In this example, $V_0 = \{a, b, c, d, e\}$, $V_1 = \{f, g, h, i, j\}$ and $V_2 = \{k, l, m, n\}$.}
\label{fig: double_layer}
\end{figure}
The {\em level} of a vertex of $G$ is defined as follows~\cite{Baker94}:
a vertex has level $0$ if it is on the infinite face of $G$;
a vertex has level $i$ if it is on the infinite face after deleting all the vertices of levels less than $i$.
Let $V_i$ be the set of vertices of level $i$.
Let $E_i$ be the edge set of $G$ in which each edge has both endpoints in level $i$.
Let $E_{i,i+1}$ be the edge set of $G$ where each edge has one endpoint in level $i$ and one endpoint in level $i+1$.
See Figure~\ref{fig: double_layer} as an example.
\ifFull
Then we have the following observations.
\begin{observation}\label{obs: circuit}
For any level $i \ge 0$, the boundary of any non-trivial two-edge connected component in $G[\cup_{j \ge i} V_j]$ is a maximal circuit in $\partial (G[V_i])$.
\end{observation}
\begin{observation}\label{obs: cycle}
For any level $i \ge 0$, the boundary of any non-trivial biconnected component in $G[\cup_{j \ge i} V_j]$ is a simple cycle in $\partial (G[V_i])$.
\end{observation}
\fi

For any $i\ge 0$, we define the $i$th {\em double layer} $$D_i = E_{i-1, i} \cup E_i \cup E_{i, i+1} \cup E_{i+1}$$ as the set of edges in $G[V_{i-1}\cup V_i \cup V_{i+1}] \setminus E_{i-1}$. See Figure~\ref{fig: D}
\begin{figure}
\centering
\includegraphics[scale = 1]{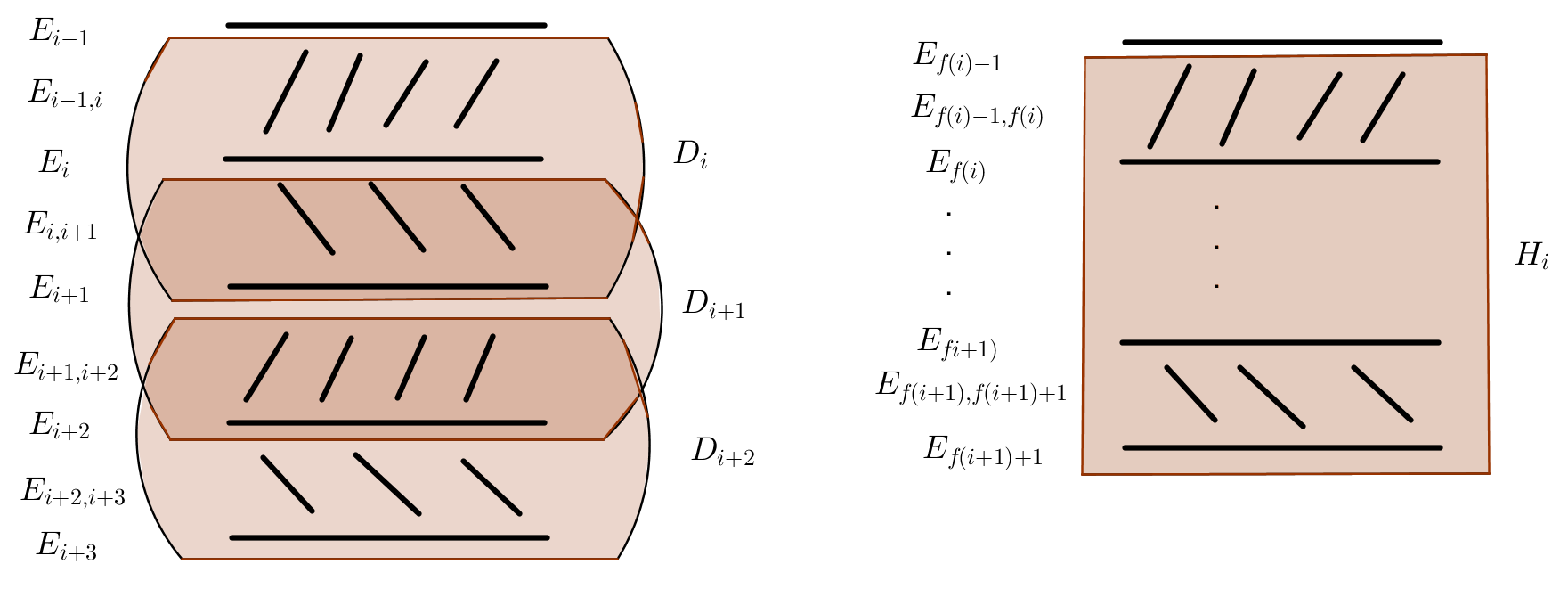}
\caption{The horizontal lines represent edge set in the same level and slashes and counter slashes represent the edge set between two adjacent levels. Left: there are three double layers: $D_i$, $D_{i+1}$ and $D_{i+2}$ represented by the shaded regions. Right: $G_i$ contains all edges in this figure, but $H_i$, represented by shaded region, does not contain $E_{f(i) - 1}$.}
\label{fig: D}
\end{figure}

Let $k$ be a constant that depends on $\epsilon$. For $j = 0, 1, \dots, k-1$, let $R_j = \cup_{i\mod k = j} D_{i}$. 
Let $t = \argmin_j |R_j|$ and $R = R_t$.
Since $\sum_{j=0}^{k-1} |R_j| \le 2|G|$, we have the following upper bound for the size of $R$.
\begin{align}
|R| \le 2/k\cdot|G| \label{equ: R}
\end{align}

Let $f(i) = ik-k+t$ for integer $i \ge 0$. 
If $ik-k+t < 0$ for any $i$, we let $f(i) = 0$.
Let $G_i = G[\cup_{f(i)-1 \le j \le f(i+1)+1} V_j]$ be the subgraph of $G$ induced by vertices in level $[f(i)-1, f(i+1)+1]$ and $H_i = G_i \setminus E_{f(i)-1}$ be a subgraph of $G_i$. 
See Figure~\ref{fig: D}.
Note that $H_i$ contains exactly the edges of double layers $D_{f(i)}$ through $D_{f(i+1)}$.
Therefore, so long as $k\ge 2$, we have 
$H_i \cap H_{i+1} = D_{f(i+1)} \subseteq R$, and $H_i \cap H_j = \emptyset$ for any $j \ne i$ and $|i-j| \ge 2$.
So for any $j \ne i$ we have 
\begin{align}
(H_i \setminus R) \cap (H_j \setminus R) = \emptyset. \label{equ: empty}
\end{align}

For each $i\ge 0$ and each maximal circuit $C_a$ in $\partial (G[V_{f(i)}])$, we construct a graph $H_i^a$, called a {\em 3EC slice}, from $G$ as follows. (See Figure~\ref{fig: Hi}.)
\begin{figure}
\centering
\includegraphics[scale = 1]{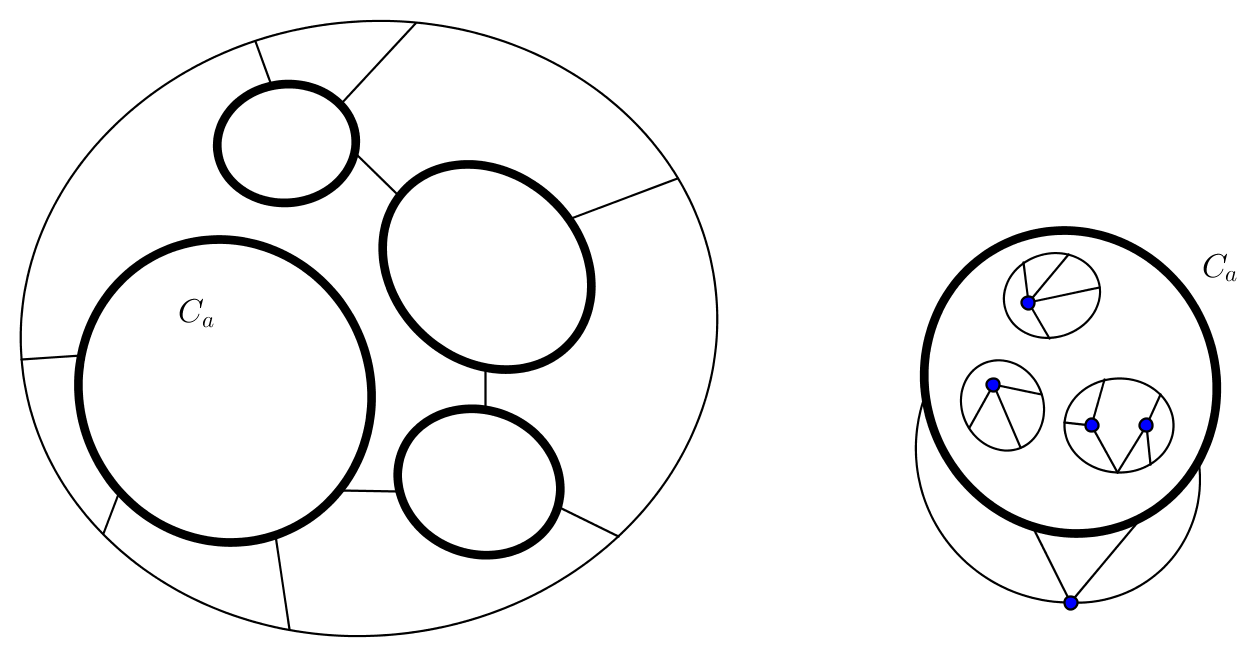}
\caption{Example for the construction of $H_i^a$. Left: a component of $G_i$. The bold cycles represent maximal circuits in $\partial(G[V_{f(i)}])$. Right: an example of $H_i^a$. The nodes represent the contracted nodes. The cycles inside of $C_a$ must belong to $E_{f(i+1)}$.}
\label{fig: Hi}
\end{figure}
Let $U$ be the subset of vertices of $H_i \setminus (V_{f(i)-1} \cup V_{f(i+1)+1})$ that are enclosed by $C_a$.
We contract each connected component of $G\setminus U$ into a node. 
After all the contractions, we delete self-loops and additional parallel edges if there are more than three parallel edges between any pair of vertices.
The resulting graph is the 3EC slice $H_i^a$.
We call these contracted vertices {\em nodes} to distinguish them from the original vertices of $G$.
We call a contracted node {\em inner} if it is obtained by contracting a component that is enclosed by $C_a$; otherwise it is {\em outer}.
Note that a 3EC slice is still planar, and two 3EC slices only share edges in a double layer: the common edges of two 3EC slices $H_i^a$ and $H_{i+1}^b$ must be in the set $E_{f(i+1)-1, f(i+1)} \cup E_{f(i+1)} \cup E_{f(i+1), f(i+1)+1}$, while the common edges of $H_i^a$ and $H_i^c$ must be in the set $E_{f(i)}$.
In the similar way, we can construct a simple graph $H_i^a$, called {\em 3VC slice}, for each $i\ge 0$ and each simple cycle $C_a$ in $\partial (G[V_{f(i)}])$.

\ifFull
\begin{remark}
There can be a 3EC slice $H_i^a$ containing only two vertices in $V_{f(i)}$. Then the slice must contain at least two parallel edges between the two vertices in $V_{f(i)}$.
But any 3EC slice $H_i^a$ cannot contain only one vertex in $V_{f(i)}$ since we define 3EC slice based on a maximal circuit and we assume there is no self-loop in $G$.
Similarly, any 3VC slice $H_i^a$ contains at least three vertices in $V_{f(i)}$. 
\end{remark}
\fi

\begin{lemma}\label{lem: outer}
If $G$ is two-edge connected (biconnected), then each 3EC (3VC) slice obtained from $G$ has at most one outer node.
\end{lemma}
\begin{proof}
\ifFull
We first prove the following claims, and then by these claims we prove the lemma.
\fi
\begin{claim}\label{clm: out}
For any $l \ge 0$, subgraph $G[\cup_{0\le j \le l} V_j]$ is connected.
\end{claim}
\ifFull
\begin{proof}
We prove by induction on $l$ that subgraph $G[\cup_{0\le j \le l} V_j]$ is connected for any $l \ge 0$.
The base case is $l = 0$. Since $V_0$ is the set of vertices on the boundary of $G$, and since $G$ is connected, subgraph $G[V_0]$ is connected.
Assume subgraph $G[\cup_{0 \le j \le l} V_j]$ is connected for $l \ge 0$.
Then we claim subgraph $G[\cup_{0 \le j \le l+1} V_j]$ is connected.
This is because for each connected component $X$ of $G[V_{l+1}]$, there exists at least one edge between $X$ and $G[V_{l}]$, 
otherwise graph $G$ cannot be connected.
Since subgraph $G[\cup_{0\le j \le l} V_j]$ is connected, we have $G[\cup_{0 \le j \le l+1} V_j]$ is connected.
\end{proof}
\fi
\begin{claim}\label{clm: 3EC}
If $G$ is two-edge connected, then for any two distinct maximal circuits $C_a$ and $C_b$ in $\partial (G[V_{l}])$, there is a path between $C_b$ and $G[V_{l-1}]$ that is vertex disjoint from $C_a$.
\end{claim}
\begin{proof}
Note that $C_a$ and $C_b$ are vertex-disjoint, otherwise $C_a$ is not maximal. 
We argue that there cannot be two edge-disjoint paths between $C_a$ and $C_b$ in $G[V_{l}]$. 
If there are such two edge-disjoint paths, say $P_1$ and $P_2$, 
then $C_a$ cannot be a maximal circuit in $\partial (G[V_{l}])$: if $P_1$ and $P_2$ have the same endpoint in $C_a$, then $C_a$ is not maximal; otherwise there is some edge of $C_a$ that cannot be in $\partial (G[V_{l}])$.
So we know that any vertex in $C_a$ and any vertex in $C_b$ cannot be two-edge connected in $G[V_{l}]$.
Since $G$ is two-edge connected and since $G[V_{l-1}]$ must be outside of $C_a$ and $C_b$, vertices in $C_a$ and those in $C_b$ must be connected through $G[V_{l-1}]$.
Therefore, there exists a path from $C_b$ to $G[V_{l-1}]$ that does not contain any vertex in $C_a$.
\end{proof}
Similarly, we can obtain the following claim.
\begin{claim}\label{clm: 3VC}
If $G$ is biconnected, then for any two distinct simple cycles $C_a$ and $C_b$ in $\partial (G[V_{l}])$, there is a path between $C_b$ and $G[V_{l-1}]$ that is vertex disjoint from $C_a$.
\end{claim}

Now we prove the lemma.
Let $H$ be a 3EC slice based on some maximal circuit $C_a$ in $\partial (G[V_{l}])$ for some $l\ge 0$.
Let $W = \cup_{0 \le j <l} V_j$ be the set of all vertices of $G$ that have levels less than $l$, and $Q$ be a two-edge connected component in $G[V_{l}]$ disjoint from $H$.
Then the boundary of $Q$ is a maximal circuit $C$ in $\partial (G[V_{l}])$.
Note that $Q$ could be trivial and then $C$ is also trivial.
Since $G$ is connected, each simple cycle must enclose a connected subgraph of $G$.
So circuit $C$ must enclose a connected subgraph of $G$.
By Claim~\ref{clm: 3EC}, there is a path between $C$ and $G[V_{l-1}]$ disjoint from $C_a$.
Since $G[\cup_{0\le j < l}V_j]$ is connected by Claim~\ref{clm: out}, the set of vertices that are not enclosed by $C_a$ induces a connected subgraph of $G$, giving the lemma for $H$.
For 3VC slice, we can obtain the lemma in the same way by Claim~\ref{clm: out} and Claim~\ref{clm: 3VC}.  
\end{proof}

Using this lemma, we show how to construct all the 3EC slices in linear time.
First we compute the levels of all vertices in linear time by using an appropriate representation of the planar embedding such as that used by Lipton and Tarjan~\cite{LT79}.
We construct all 3EC slices $H_i^a$ from $G_i$ in $O(|V(G_i)|)$ time.
We first contract all the edges between vertices of level $f(i+1)+1$.
Next, we identify all two-edge connected components in $G[V_{f(i)}]$, which can be done in linear time by finding all the edge cuts by the result of Tarjan~\cite{Tarjan74}.
Each such component contains a maximal circuit in $\partial (G[V_{f(i)}])$.
Based on these two-edge connected components of $G[V_{f(i)}]$, we could identify 
$V(H_i^a) \setminus \{ r_i^a \}$ for all 3EC slices $H_i^a$ in $O(|V(G_i)|)$ time, where $r_i^a$ is the outer contracted node for $H_i^a$.
This is because the inner contracted nodes of a 3EC slice $H_i^a$ is the same as those contracted in $G_i$ if they are enclosed by $C_a$.
Then for each 3EC slice $H_i^a$ we add the outer node $r_i^a$, and for each vertex $u \in V(C_a)$ we add an edge between $r_i^a$ and $u$ if there is an edge between $u$ and some vertex $v$ that is not enclosed by $C_a$. 
To add those edges, we only need to travel all the edges of subgraph $G_i[V_{f(i)-1} \cup V_{f(i)}]$.
Since all these steps run in $O(|V(G_i)|$ time, and since $\sum_{i\ge 0} |V(G_i)| = O(|V(G)|)$,
we can obtain the following lemma.
\begin{lemma}\label{lem: 3EC_slice}
All 3EC slices can be constructed in $O(|V(G)|)$ time.
\end{lemma}
Since we can compute all the biconnected components in $G[V_{f(i)}]$ in linear time based on depth-first search by the result of Hopcroft and Tarjan~\cite{HT73}, we can obtain a similar lemma for 3VC sllices in a similar way.
\begin{lemma}\label{lem: 3VC_slice}
All 3VC slices can be constructed in $O(|V(G)|)$ time.
\end{lemma}

We review the definition of {\em branchwidth} given by Seymour and Thomas~\cite{ST94}. A {\em branch decomposition} of a graph $G$ is a hierarchical clustering of its edge set. We represent this hierarchy by a binary tree, called the {\em decomposition tree}, where the leaves are in bijection with the edges of the original graph. 
If we delete an edge $\alpha$ of this decomposition tree, the edge set of the original graph is partitioned into two parts $E_{\alpha}$ and $E(G) \setminus E_{\alpha}$ according to the leaves of the two subtrees. 
The set of vertices in common between the two subgraphs induced by $E_{\alpha}$ and $E(G) \setminus E_{\alpha}$ is called the {\em separator} corresponding to $\alpha$ in the decomposition.
The {\em width} of the decomposition is the maximum size of the separator in that decomposition, and the {\em branchwidth} of $G$ is the minimum width of any branch decomposition of $G$.
\ifFull
We borrow the following lemmas from Klein and Mozes~\cite{planar}, which are helpful in bounding the branchwidth of our graphs.
\begin{lemma}\label{lem: contracting}{\rm (Lemma 14.5.1~\cite{planar})}
Deleting or contracting edges does not increase the branchwidth of a graph.
\end{lemma}
\begin{lemma}\label{lem: bounded_bw}{\rm (Lemma 14.6.1~\cite{planar} rewritten)}
There is a linear-time algorithm that, given a planar embedded graph $G$, returns a branch-decomposition whose width is at most twice of the depth of a rooted spanning tree of $G$.
\end{lemma}
\fi
\begin{lemma}\label{lem: bw}
If $G$ is two-edge connected (biconnected), the branchwidth of any 3EC (3VC) slice is $O(k)$.
\end{lemma}
\ifFull
\begin{proof}
We prove this lemma for 3EC slices when $G$ is two-edge connected; by the same proof we can obtain the lemma for 3VC slices when $G$ is biconnected.
Let $H_i^a$ be a 3EC slice.
By Lemma~\ref{lem: outer}, there is at most one outer contracted node $r$ for $H_i^a$.
Let the level of $r$ be $f(i)-1$, and the level of all inner contracted nodes be $f(i+1)+1$.
Now we add edges to ensure that every vertex of level $l$ has a neighbor of level $l-1$ for each $f(i) \le l\le f(i)+1$, while maintaining planarity.
Call the resulting graph $K_i^a$.
Then the branchwidth of $H_i^a$ is no more than that of $K_i^a$ by Lemma~\ref{lem: contracting}.
Now we can find a breadth-first-tree of $K_i^a$ rooted at $r$ that has depth at most $k+3$. 
By Lemma~\ref{lem: bounded_bw}, the branchwidth of $K_i^a$ is $O(k)$ and that of $H_i^a$ is at most $O(k)$.  
\end{proof}
\fi

\section{PTAS for $3$-ECSS}\label{sec: EC}
In this section, we prove Theorem~\ref{thm: EC}.
W.l.o.g. we assume $G$ has at most three parallel edges between any pair of vertices.
Then $G$ is our spanner.
Let $\opt(G)$ be an optimal solution for $G$.
Since each vertex in $\opt(G)$ has degree at least three,
we have 
\begin{align}
2|\opt(G)| \ge 3 |V(G)|. \label{equ: opt1}
\end{align}
If $G$ is simple, then by planarity the number of edges is at most three times of the number of vertices.
Since there are at most three parallel edges between any pair of vertices, we have 
\begin{align}
|G| \le 9 |V(G)|.  \label{equ: opt2}
\end{align}
Combining~(\ref{equ: opt1}) and~(\ref{equ: opt2}), we have $|G| \le 6 |\opt(G)|$. 

In this section, we only consider 3EC slices. So when we say slice, we mean 3EC slice in this section.
We construct all the slices from $G$.
By~(\ref{equ: R}), we have the following 
\begin{align}
|R| \le 2/k \cdot |G| \le 12/k \cdot |\opt(G)|. \label{equ: EC1}
\end{align}

We borrow the following lemma from Nagamochi and Ibaraki~\cite{NI92}.
\begin{lemma}\label{lem: contraction}{\rm (Lemma 4.1 (2)~\cite{NI92} rewritten)}
Let $G$ be a $k$-edge connected graph with more than 2 vertices. Then after contracting any edge in $G$, the resulting graph is still $k$-edge connected.
\end{lemma}
Recall that our slices are obtained from $G$ by contractions and deletions of self-loops. By the above lemma, we have the following lemma.
\begin{lemma}\label{lem: EC1}
If $G$ is three-edge connected, then any slice is three-edge connected.
\end{lemma}
Since we can include all the edges in shared double layers, they are ``free'' to us.
So we would like to include those edges as many as possible in the solution for each slice.
This can be achieved by defining an edge-weight function $w$ for each slice $H_i^a$: assign weight $0$ to edges in $D_{f(i)} \cup D_{f(i+1)}$ and weight $1$ to other edges.
By Lemma~\ref{lem: EC1}, any slice is three-edge connected.
We solve the minimum-weight 3-ECSS problem on $H_i^a$ in linear time by Theorem~\ref{thm: dp}.
Let $Sol(H_i^a)$ be a feasible solution for the minimum-weight 3-ECSS problem on $H_i^a$. 
Then it is also a feasible solution for 3-ECSS on $H_i^a$.
Let $\opt_w(H_i^a)$ be an optimal solution for the minimum-weight 3-ECSS problem on $H_i^a$.
\ifFull
Then we have the following observation.
\begin{observation}\label{obs: weight}
The weight of any solution $Sol(H_i^a)$ is the same as the number of its common edges with $H_i^a \setminus R$, that is $$w(Sol(H_i^a)) = |Sol(H_i^a) \cap (H_i^a \setminus R) |.$$
\end{observation}
\fi
Let ${\cal C}_i$ be the set of all maximal circuits in $\partial(G[V_{f(i)}])$.
\ifFull
Then we have the following lemmas.
\fi
\begin{lemma}\label{lem: EC3}
For any $i\ge 0$, let $S_i = \bigcup_{ C_a \in {\cal C}_i} \opt_w(H_i^a)$. Then
we can bound the number of edges in $S_i$ by the following inequality $$|S_i| \le |\opt(G) \cap (H_i \setminus R)| + |D_{f(i)}| + |D_{f(i+1)}|.$$
\end{lemma}
\ifFull
\begin{proof}
We show that $\opt(G) \cap H_i^a$ is a feasible solution for the minimum-weight 3-ECSS problem on $H_i^a$, and then we bound the size of $S_i$.
Let $Y_i^a$ be the set of vertices of $H_i^a$ that are not contracted nodes.
We first contract connected components of $\opt(G) \setminus Y_i^a$ just as constructing $H_i^a$ from $G$.
Then we need to identify any two contracted nodes, if their corresponding components in $\opt(G)$ are in the same connected component in $G\setminus Y_i^a$. See Figure~\ref{fig: identify}.
\begin{figure}
\centering
\includegraphics[scale = 1]{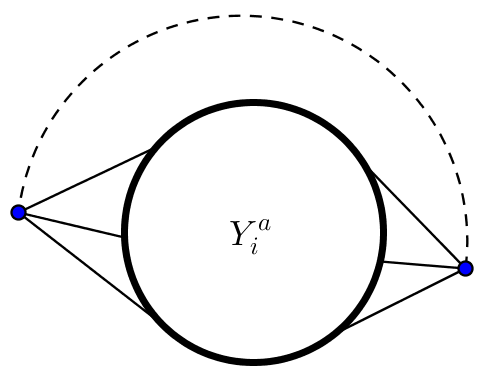}
\caption{The bold cycle encloses $Y_i^a$. The dashed edge is in $G$ but not in $\opt(G)$. Its two endpoints will be identified, since the dashed edge will be contracted to obtain $H_i^a$ but it will not be contracted when contracting connected components of $\opt(G) \setminus Y_i^a$. }
\label{fig: identify}
\end{figure}
Finally, we delete all the self-loops and extra parallel edges if there are more than three parallel edges between any two vertices.
The resulting graph is a subgraph of $\opt(G) \cap H_i^a$ and spans $V(H_i^a)$.
Since identifying two nodes maintains edge-connectivity, and since contractions also maintain edge-connectivity by Lemma~\ref{lem: contraction}, the resulting graph is three-edge connected. 
So $\opt(G) \cap H_i^a$ is a feasible solution for minimum-weighted 3-ECSS problem on $H_i^a$.
Then 
by the optimality of $\opt_w(H_i^a)$, we have $w(\opt_w(H_i^a)) \le w(\opt(G) \cap H_i^a)$.
And by Observation~\ref{obs: weight}, we have 
\begin{align}
|\opt_w(H_i^a) \cap (H_i^a \setminus R)| \le |(\opt(G) \cap H_i^a) \cap (H_i^a \setminus R) | = |\opt(G) \cap (H_i^a \setminus R)|.
\label{equ: EC11}
\end{align}

Note that for any slice $H_i^a$, we have $E(H_i^a) \subseteq E(H_i)$ and $(H_i^a \cap R) \subseteq (H_i \cap R) \subseteq (D_{f(i)} \cup D_{f(i+1)})$.
Since for distinct (vertex-disjoint) maximal circuits $C_a$ and $C_b$ in ${\cal C}_i$, subgraphs $H_i^a \setminus R$ and $H_i^b \setminus R$ are vertex-disjoint,
we have the following equalities. 
\begin{align}
H_i \setminus R = \bigcup\nolimits_{ C_a \in {\cal C}_i } (H_i^a \setminus R) 
\label{equ: EC2}
\end{align}
\begin{align}
S_i \cap (H_i \setminus R) = \bigcup\nolimits_{ C_a \in {\cal C}_i } (\opt_w(H_i^a) \cap (H_i^a \setminus R))
\label{equ: EC3}
\end{align} 
Then 
$$\begin{array}{lll}
 |S_i \cap (H_i \setminus R)| & = \left| \bigcup_{ C_a \in {\cal C}_i } (\opt_w(H_i^a) \cap (H_i^a \setminus R)) \right|  & \mbox{by~(\ref{equ: EC3})} \\
 & \le \sum_{ C_a \in {\cal C}_i } |\opt_w(H_i^a) \cap (H_i^a \setminus R) | &  \\
 & \le \sum_{ C_a \in {\cal C}_i } |\opt(G) \cap (H_i^a \setminus R) |  & \mbox{by~(\ref{equ: EC11})}\\
 & \le |\opt(G) \cap (H_i \setminus R)|. & \mbox{by~(\ref{equ: EC2})}
\end{array}
$$
So we have $\left| S_i \right| = |S_i \cap (H_i \setminus R)| + |S_i \cap (H_i \cap R) | \le |\opt(G) \cap (H_i \setminus R)| + |D_{f(i)}| + |D_{f(i+1)}|$.
\end{proof}
\fi
\begin{lemma}\label{lem: EC2}
The union $\left( \bigcup_{i\ge 0, C_a \in {\cal C}_i } Sol(H_i^a) \right) \cup R$ is a feasible solution for $G$.
\end{lemma}
\ifFull
\begin{proof}
\else
\begin{proof}[Proof (sketch)]
\fi
For any $i\ge 0$ and any maximal circuit $C_a \in {\cal C}_i$, let $M_i^a$ be the graph obtained from $H_i^a$ by uncontracting all its inner contracted nodes.
See Figure~\ref{fig: tree}.
By Lemma~\ref{lem: outer}, there is at most one outer node $r_i^a$ for each slice $H_i^a$.

Define a tree $T$ based on all the slices: each slice is a node of $T$, and two nodes $H_i^a$ and $H_{j}^b$ are adjacent if they share any edge and $|i-j| =1$. 
Root $T$ at the slice $H_0^a$, which contains the boundary of $G$. 
Let $T(H_{i}^a)$ be the subtree of $T$ that roots at slice $H_{i}^a$.
See Figure~\ref{fig: tree} as an example.
\ifFull
For each child $H_{i+1}^b$ of $H_i^a$, let $C_b$ be the boundary of $H_{i+1}^b \setminus \{ r_{i+1}^b\}$.
Then $C_b$ is the maximal circuit in ${\cal C}_{i+1}$ that is shared by $H_i^a$ and $H_{i+1}^b$.
Further, by 
\fi 
the construction of $H_i^a$, graph $M_{i+1}^b\setminus \{ r_{i+1}^b\}$ is a subgraph of $M_i^a$.
\begin{figure}
\centering
\includegraphics[scale = 1]{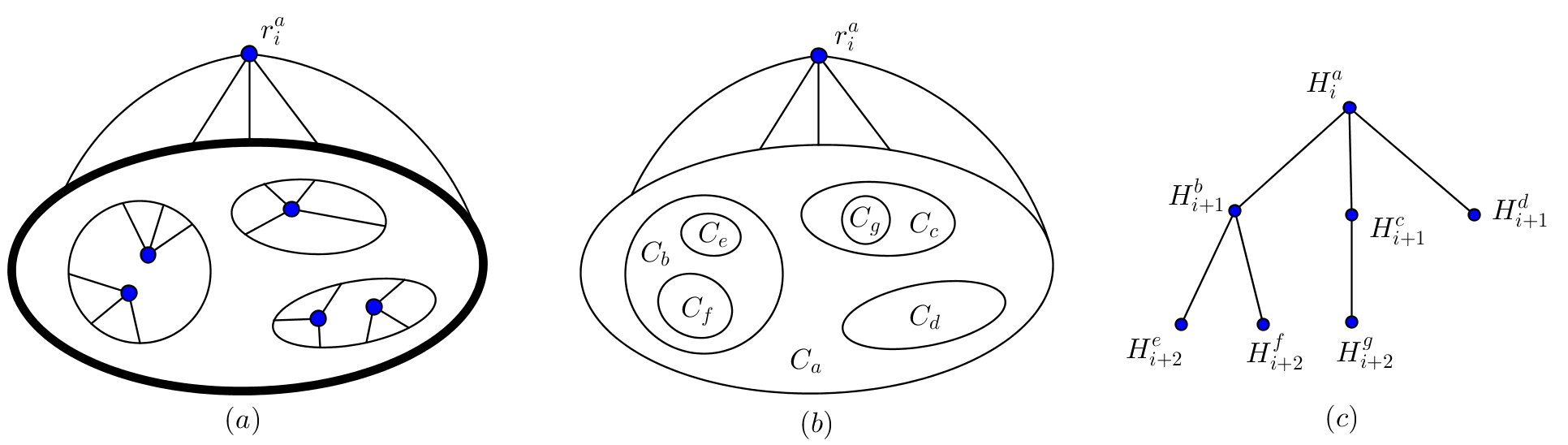}
\caption{(a) A slice $H_i^a$: the bold cycle is a maximal circuit $C_a$ in ${\cal C}_i$ and the nodes represent all the contracted nodes. (b) The graph $M_i^a$ obtained from $H_i^a$ by uncontracting inner nodes of $H_i^a$.
(c) The subtree $T(H_i^a)$.
}
\label{fig: tree}
\end{figure}

We prove the lemma by induction on this tree $T$ from leaves to root.
Assume for each child $H_{i+1}^b$ of $H_i^a$, there is a feasible solution $S^b$ for the graph $M_{i+1}^b$ such that
$S^b = \left( \bigcup_{H \in T(H_{i+1}^b)} Sol(H) \right) \cup  \left( M_{i+1}^b \cap \left(\bigcup_{j \ge i+1} D_{f(j+1)} \right) \right)$.  
We prove that there is a feasible solution $S^a$ for $M_i^a$ such that $S^a = \left( \bigcup_{H \in T(H_{i}^a)} Sol(H) \right) \cup \left( M_i^a \cap \left( \bigcup_{j \ge i} D_{f(j+1)} \right) \right)$.
For the root $H_0^a$ of $T$, we have $M_0^a \cap \left(\bigcup_{j\ge 0} D_{f(j+1)} \right) \subseteq R$, and then the lemma follows from the case $i = 0$.

The base case is that $H_i^a$ is a leaf of $T$.
When $H_i^a$ is a leaf, there is no inner contracted node in $H_i^a$ and we have $M_i^a = H_i^a$. 
So $Sol(H_i^a)$ is a feasible solution for $M_i^a$. 

\ifFull

Recall that $H_i^a$ and $H_{i+1}^b$ only share edges of $(E_{f(i+1)-1, f(i+1)} \cup E_{f(i+1)} \cup E_{f(i+1), f(i+1)+1}) \subseteq D_{f(i+1)}$ and vertices of $V_{f(i+1)}$.
Let $x$ be any inner contracted node of $H_i^a$ and $X$ be the vertex set of the connected component of $G$ corresponding to $x$.
We need the following claim.
\begin{claim}\label{clm: EC1}
If $X \subseteq M_{i+1}^b$ for some $H_{i+1}^b$, then $(S^b \cup D_{f(i+1)}) \cap G[X]$ is connected.
\end{claim}
\begin{proof}
By the construction of levels, all the vertices on the boundary of $G[X]$ have level $f(i+1)+1$.
See Figure~\ref{fig: claim}.
Then all the edges of $\partial (G[X])$ are in $E_{f(i+1)+1} \subseteq D_{f(i+1)}$.
So subgraph $(S^b \cup D_{f(i+1)}) \cap \partial(G[X])$ is connected.
Let $u$ be any vertex in $X$ and let $v$ be any vertex in $M_{i+1}^b$ that has level $f(i+1)$.
Then $v$ is not in $X$.
Since $S^b$ is a feasible solution for $M_{i+1}^b$, there exists a path from $u$ to $v$ in $S^b$. 
This path must intersect $\partial (G[X])$ by planarity.
So $u$ and any vertex on the boundary of $G[X]$ are connected in $(S^b \cup D_{f(i+1)}) \cap G[X]$, giving the claim.  
\end{proof}
\begin{figure}
\centering
\includegraphics[scale = 1]{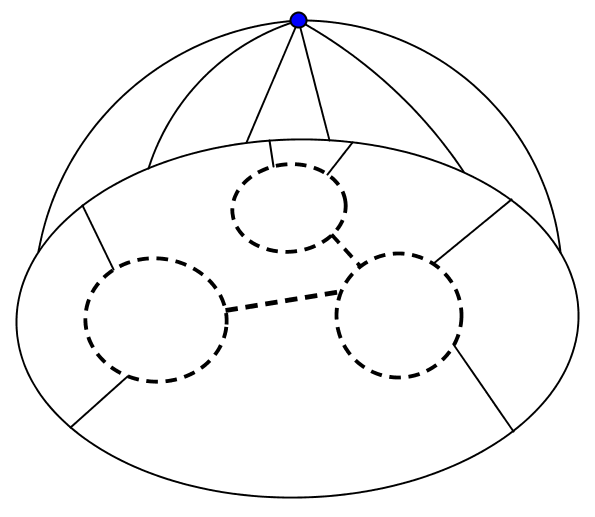}
\caption{
The dashed subgraph is the boundary of $G[X]$. All the vertices in the dashed subgraph are in level $f(i+1)+1$, and all its edges are in $E_{f(i+1)+1}$.
}
\label{fig: claim}
\end{figure}

Let $u$ and $v$ be any two vertices of $M_i^a$. 
To prove the feasibility of $S^a$, we prove $u$ and $v$ are three-edge connected in $S^a$.
Let $M = \left( \bigcup_{H_{i+1}^b \mbox{ is a child of } H_i^a } V(M_{i+1}^b \setminus \{ r_{i+1}^b \}) \right)$ and $Y_i^a = V(H_i^a) \setminus \{ \mbox{inner contracted nodes of } H_i^a \}$.
Then $V(M_i^a) = Y_i^a \cup M$.
Depending on the locations of $u$ and $v$, we have three cases.
\begin{description}
\item[Case 1: $\mathbf{u,v \in Y_i^a}$.]
Note that we could construct $S^a$ in the following way.
Initially we have $S^* = Sol(H_i^a) \cup (D_{f(i+1)} \cap H_i^a)$.
For any inner contracted node $x$ of $H_i^a$, let $X$ be the vertex set of its corresponding connected component in $G$.
Then there exists a child $H_{i+1}^b$ of $H_i^a$ such that $X \subseteq V(M_{i+1}^b)$, and we replace $x$ with $(S^b \cup D_{f(i+1)}) \cap G[X]$ in $S^*$.
We do this for all inner contracted nodes of $H_i^a$.
Finally we add some edges of $D_{f(i+1)}$ into the resulting graph such that $D_{f(i+1)} \subseteq S^*$.
Then the resulting $S^*$ is the same as $S^a$ by the definition of $S^a$.
We prove that any pair of the remaining vertices in $V(H_i^a)$ are three-edge connected during the construction.
This includes the remaining inner contracted nodes of $H_i^a$ during the process, but after all the replacements, there is no such inner contracted nodes, proving the case.

By the definition of $Sol(H_i^a)$, any pair of vertices of $H_i^a$ are three-edge connected in $Sol(H_i^a)$.
Assume after the first $k$ replacements, any pair of the remaining vertices in $V(H_i^a)$ are three-edge connected in the resulting graph $S^*$. 
Let $x$ be the next inner contracted node to be replaced, $X$ be the vertex set of its corresponding component and $S'$ be the resulting graph after replacing $x$.
Let $H_{i+1}^b$ be the child of $H_i^a$ such that $X \subseteq V(M_{i+1}^b)$.
Let $C$ be the simple cycle in $\partial (M_{i+1}^b \setminus \{ r_{i+1}^b \} )$ that encloses $X$.
Then all vertices of $C$ have level $f(i+1)$ and are shared by $H_i^a$ and $H_{i+1}^b$.
Further, $C \subseteq H_i^a \cap D_{f(i+1)} \subseteq S'$.
Let $u$ and $v$ be any two remaining vertices of $V(H_i^a)$.
There are three edge-disjoint $u$-to-$x$ paths and three edge-disjoint $v$-to-$x$ paths in $S^*$, all of which must intersect $C$.
So there exist three edge-disjoint $u$-to-$X$ paths and three edge-disjoint $v$-to-$X$ paths in $S'$.
Now we delete two edges in $S'$.
If these two edges are not both in $C$, then the vertices of $C$ are still connected.
Then one remaining $u$-to-$C$ path and one remaining $v$-to-$C$ path together with the rest of $C$ witness the connectivity between $u$ and $v$.
If the two deleted edges are both in $C$, then there exist one $u$-to-$X$ path and one $v$-to-$X$ path after the deletion.
By Claim~\ref{clm: EC1}, subgraph $(S^b \cup D_{f(i+1)}) \cap G[X]$ is connected. So all vertices of $X$ are connected in $S'$.
Then $u$ and $v$ are connected after the deletion.
Finally, after replacing all the inner contracted nodes, we only add edges of $D_{f(i+1)}$ into $S^*$, which will not break three-edge-connectivity between any pair of vertices.
This finishes the proof of Case 1.

\item[Case 2: $\mathbf{u, v \in M}$.] 
Let $M_{i+1}^{b_1}$ be the graph contains $u$ and $M_{i+1}^{b_2}$ be the graph contains $v$. (The two graphs could be identical.) 
Let $C_u$ (resp. $C_v$) be the simple cycle in $\partial (M_{i+1}^{b_1} \setminus \{ r_{i+1}^{b_1} \})$ (resp. $\partial (M_{i+1}^{b_2} \setminus \{ r_{i+1}^{b_2} \})$) that enclose $u$ (resp. $v$). (The two cycles $C_u$ and $C_v$ could be identical.)
Since $S^{b_1}$ is three-edge connected, there are three edge-disjoint paths from $u$ to some vertex of $C_u$ in $S^{b_1}$.
All these three paths must intersect $C_u$, so there are three edge-disjoint paths from $u$ to $C_u$ in $(S^{b_1} \setminus \{ r_{i+1}^{b_1} \}) \subseteq S^a$.
Similarly, there are three edge-disjoint paths from $v$ to $C_v$ in $(S^{b_2} \setminus \{ r_{i+1}^{b_2} \}) \subseteq S^a$.
Now we delete any two edges in $S^a$.
After the deletion, there exist one $u$-to-$w_1$ path and one $v$-to-$w_2$ path where $w_1 \in C_u$ and $w_2 \in C_v$.
Since all vertices in $V (C_u \cup C_v)$ have level $f(i+1)$ and are in $Y_i^a$, they are three-edge connected in $S^a$ by Case 1.
This means there exists a path from $w_1$ to $w_2$ after the deletion.
Therefore, $u$ and $v$ are connected after deleting any two edges in $S^a$, giving the three-edge-connectivity.

\item[Case 3: $\mathbf{u \in Y_i^a}$ {\bf and} $\mathbf{v \in M}$.]
Let $M_{i+1}^b$ be the graph containing $v$.
Then there is a vertex $w$ in $Y_i^a \cap (M_{i+1}^b \setminus \{ r_{i+1}^b \})$.
By Case 1, vertices $u$ and $w$ are three-edge connected, and by Case 2, vertices $v$ and $w$ are three-edge connected.
Then vertices $u$ and $v$ are three-edge connected by the transitivity of three-edge-connectivity.
\end{description}
This completes the proof of Lemma~\ref{lem: EC2}.
\fi
\end{proof}

\begin{proof}[Proof of Theorem~\ref{thm: EC}]
We first prove correctness of our algorithm, and then prove its running time.
By Lemma~\ref{lem: EC2}, $S = \left( \bigcup_{i\ge 0, C_a \in {\cal C}_i} \opt_w(H_i^a) \right) \cup R$ is a feasible solution. 
Thus
$$
\begin{array}{llll}
|S| 
 & \le \left|\left( \bigcup_{i\ge 0, C_a \in {\cal C}_i } \opt_w(H_i^a) \right) \right| + |R|  & \mbox{} \\
 & \le \sum_{i\ge 0} \left| \bigcup_{ C_a \in {\cal C}_i} \opt_w(H_i^a)   \right| +  |R|  & \mbox{} \\
 & \le \sum_{i\ge 0}\left(\left| \opt(G)\cap (H_i\setminus R) \right| +|D_{f(i)}| +|D_{f(i+1)}|  \right) + |R|  & \mbox{by Lemma~\ref{lem: EC3}} \\
 & \le \sum_{i\ge 0}\left| \opt(G)\cap (H_i\setminus R) \right| +|R| +|R| + |R|  & \mbox{} \\
 & \le |\opt(G)| + 3|R|  & \mbox{by~(\ref{equ: empty})} \\
 & \le |\opt(G)| + 36/k\cdot |\opt(G)|  & \mbox{by~(\ref{equ: EC1})} \\
 & \le (1+36/k)|\opt(G)|  & \mbox{} \\
\end{array}
$$
Let $k = 36/\epsilon$, and then we obtain $|S| \le (1+\epsilon) |\opt(G)|$.

Let $n = |V(G)|$ be the number of vertices of graph $G$.
We could find $R$ and construct all slices in $O(n)$ time by Lemma~\ref{lem: 3EC_slice}.  
By Lemma~\ref{lem: bw}, each slice has branchwidth $O(k)$.
So by Theorem~\ref{thm: dp}, we could solve the minimum-weight 3-ECSS on each slice in linear time for fixed $k$.
Based on those optimal solutions for all slices, we could construct our solution in $O(n)$ time.
Therefore, our algorithm runs in $O(n)$ time.
\end{proof}

\section{PTAS for $3$-VCSS}\label{sec: VC}
In this section, we prove Theorem~\ref{thm: VC}.
W.l.o.g. assume $G$ is simple.
Then $G$ is our spanner.
Let $\opt(G)$ be an optimal solution for $G$. 
Since $G$ is simple and planar, we have $|G| \le 3 |V(G)|$.  
Then by (\ref{equ: opt1}) we have $|G| \le 2|\opt(G)|$.
In this section, we only consider 3VC slices. So in the following, we simplify 3VC slice to slice.
We first construct slices from $G$.
By~(\ref{equ: R}), we have the following 
\begin{align}
|R| \le 2/k \cdot |G| \le 4/k \cdot |\opt(G)|. \label{equ: VC1}
\end{align}

Similar to 3-ECSS, we want to solve a minimum-weight 3-VCSS problem on each slice. But before defining the weights for this problem on each slice, we first need to show any slice is triconnected.
The following lemma is proved by Vo~\cite{Vo83}, we provide a proof for completeness. 
\begin{lemma}{\rm (\cite{Vo83})} \label{lem: component}
Let $C$ be a simple cycle of $G$ that separates $G\setminus C$ into two parts: $A$ and $B$. Let $H$ be any connected component of $A$.
If $G$ is triconnected, then $G/H$, the graph obtained from $G$ by contracting $H$, is triconnected.
\end{lemma}
\begin{proof}
Let $x$ be the contracted node of $G/H$.
Then $x$ and any other vertex of $G/H$ are triconnected since $G$ is triconnected. 
Let $u$ and $v$ be any two vertices of $G/H$ distinct from $x$. 
To prove the lemma, we show $u$ and $v$ are triconnected.
Since $x$ and $u$ are triconnected, there are three vertex-disjoint paths between $u$ and $x$.
Note that all the three paths must intersect cycle $C$ since $V(C)$ form a cut for $x$ and all the other vertices in $G/H$.
Similarly, there are three vertex-disjoint paths between $v$ and $x$, all of which intersect cycle $C$.
Now we delete any two vertices different from $u$ and $v$ in $G/H$.
If the two deleted vertices are both in $C$, then there 
exist one $u$-to-$x$ path and one $v$-to-$x$ path after the deletion, which witness the connectivity between $u$ and $v$.
If the two deleted vertices are not both in $C$, the remaining vertices in $C$ are connected and then the remaining $u$-to-$C$ path and the remaining $v$-to-$C$ path together with the rest of edges in $C$ witness the connectivity between $u$ and $v$.
So $u$ and $v$ are triconnected.
\end{proof}
\ifFull
By the same proof, we can obtain the following lemma.
\begin{lemma}\label{lem: identification}
Let $C$ be a simple cycle of $G$.
Let $u$ and $v$ be two vertices of $G\setminus C$ whose neighbors in $G$ are all in $C$.
Then if $G$ is triconnected, the graph obtained from $G$ by identifying $u$ and $v$ is triconnected.
\end{lemma} 
\fi
Let ${\cal C}_i$ be the set of all simple cycles in $\partial (G[V_{f(i)}])$. Then we have the following lemma.
\begin{lemma}\label{lem: VC1}
For any $i\ge 0$ and any simple cycle $C_a \in {\cal C}_i$, the slice $H_i^a$ is triconnected.
\end{lemma}
\begin{proof} 
Let $Y_i^a$ be the set of vertices of $H_i^a$ that are not contracted nodes.
We could obtain $H_i^a$ by contracting each connected component of $G\setminus Y_i^a$ into a node.
Each time we contract a connected component $H$ of $G\setminus Y_i^a$, 
there is a simple cycle $C$ that separates $H$ and other vertices: if $H$ is outside of $C_a$, then $C = C_a$; otherwise $C$ is some simple cycle in ${\cal C}_{i+1}$ that encloses $H$.
Then by Lemma~\ref{lem: component} the resulting graph is still triconnected after each contraction.
Therefore, the final resulting graph $H_i^a$ is triconnected.
\end{proof}

Now we define the edge-weight function $w$ on a slice $H_i^a$:
we assign weight 0 to edges in $H_i^a \cap (D_{f(i)} \cup D_{f(i+1)})$ and weight 1 to other edges.
Then we solve the minimum-weight 3-VCSS problem on slice $H_i^a$ by Theorem~\ref{thm: dp}.
Let $Sol(H_i^a)$ be a feasible solution for the minimum-weight 3-VCSS problem on $H_i^a$.
Then it is also a feasible solution for 3-ECSS on $H_i^a$. 
Let $\opt_w(H_i^a)$ be an optimal solution for this problem on $H_i^a$.
Then we can prove the following two lemmas, whose proofs follow the same outlines of the proofs of Lemmas~\ref{lem: EC3} and~\ref{lem: EC2} respectively.
\begin{lemma}\label{lem: VC3}
For any $i \ge 0$, let $S_i = \bigcup_{C_a \in {\cal C}_i} \opt_w(H_i^a)$.
Then we can bound the number of edges in $S_i$ by the following inequality
$$\left| 
S_i
\right| \le |\opt(G) \cap (H_i \setminus R)| + |D_{f(i)}| + |D_{f(i+1)}|.$$
\end{lemma}
\ifFull
\begin{proof}
We first show $\opt(G) \cap H_i^a$ is a feasible solution for minimum-weighted 3-VCSS problem on any slice $H_i^a$.
Let $Y_i^a$ be the set of vertices of $H_i^a$ that are not contracted nodes.
We first contract each component of $\opt(G)\setminus Y_i^a$ into a node.
The resulting graph after each contraction is still triconnected by Lemma~\ref{lem: component}.
After all the contractions, we identify any two contracted nodes $x_1$ and $x_2$ if their corresponding components in $\opt(G)$ are connected in $G\setminus Y_i^a$.
This implies there exists a simple cycle $C$ in ${\cal C}_i$ or ${\cal C}_{i+1}$ such that all neighbors of $x_1$ and $x_2$ are in $C$.
So by Lemma~\ref{lem: identification} the resulting graph after each identification is also triconnected.
Finally we delete parallel edges and self-loops if possible.
After identifying all possible nodes, the resulting graph has the same vertex set as $H_i^a$ and is triconnected.
Since the resulting graph is a subgraph of $\opt(G) \cap H_i^a$, we know $\opt(G) \cap H_i^a$ is a feasible solution for minimum-weighted 3-VCSS problem on $H_i^a$.

Note that for any slice $H_i^a$, we have $(H_i^a \cap R) \subseteq (H_i \cap R) \subseteq (D_{f(i)} \cup D_{f(i+1)})$.
By the optimality of $\opt_w(H_i^a)$, we have $w(\opt_w(H_i^a)) \le w(\opt(G) \cap H_i^a)$.
Since all the nonzero-weighted edges are in $H_i^a \setminus R$, Observation~\ref{obs: weight} still holds.
Then we have
\begin{align}
|\opt_w(H_i^a) \cap (H_i^a \setminus R)| \le |(\opt(G) \cap H_i^a) \cap (H_i^a \setminus R) | = |\opt(G) \cap (H_i^a \setminus R)|.
\label{equ: 1}
\end{align}
Since for distinct (edge-disjoint) simple cycles $C_a$ and $C_b$ in ${\cal C}_i$, subgraphs $H_i^a \setminus R$ and $H_i^b \setminus R$ are vertex-disjoint,
we have the following equalities. 
\begin{align}
H_i \setminus R = \bigcup\nolimits_{ C_a \in {\cal C}_i} (H_i^a \setminus R) 
\label{equ: 2}
\end{align}
\begin{align}
S_i \cap (H_i \setminus R) = \bigcup\nolimits_{ C_a \in {\cal C}_i } (\opt_w(H_i^a) \cap (H_i^a \setminus R))
\label{equ: 3}
\end{align} 
Then 
$$\begin{array}{lll}
 |S_i \cap (H_i \setminus R)| & = \left| \bigcup_{ C_a \in {\cal C}_i } (\opt_w(H_i^a) \cap (H_i^a \setminus R)) \right|  & \mbox{by~(\ref{equ: 3})} \\
 & \le \sum_{ C_a \in {\cal C}_i } |\opt_w(H_i^a) \cap (H_i^a \setminus R) | &  \\
 & \le \sum_{ C_a \in {\cal C}_i } |\opt(G) \cap (H_i^a \setminus R) |  & \mbox{by~(\ref{equ: 1})}\\
 & \le |\opt(G) \cap (H_i \setminus R)|. & \mbox{by~(\ref{equ: 2})}
\end{array}
$$
So we have $\left| S_i \right| = |S_i \cap (H_i \setminus R)| + |S_i \cap (H_i \cap R) | \le |\opt(G) \cap (H_i \setminus R)| + |D_{f(i)}| + |D_{f(i+1)}|$.
\end{proof}
\fi

\begin{lemma}\label{lem: VC2}
The union $\left( \bigcup_{i\ge 0, C_a \in {\cal C}_i} Sol(H_i^a) \right) \cup  R$ is a feasible solution for $G$.
\end{lemma}
\ifFull
\begin{proof}
For any $i\ge 0$ and any simple cycle $C_a \in {\cal C}_i$, let $M_i^a$ be the graph obtained from slice $H_i^a$ by uncontracting all the inner contracted nodes of $H_i^a$.
By Lemma~\ref{lem: outer}, there is at most one outer contracted node $r_i^a$ for any slice $H_i^a$.

Define a tree $T$ based on all the slices: each slice is a node of $T$, and two nodes $H_i^a$ and $H_{j}^b$ are adjacent if they share any edge and $|i-j| =1$.
Root $T$ at the slice $H_0^a$, which contains the boundary of $G$. 
Let $T(H_{i}^a)$ be the subtree of $T$ that roots at slice $H_{i}^a$.
For each child $H_{i+1}^b$ of $H_i^a$, let $C_b$ be the simple cycle in ${\cal C}_{i+1}$ that is shared by $H_i^a$ and $H_{i+1}^b$.
Then $C_b$ is the boundary of $H_{i+1}^b \setminus \{ r_{i+1}^b\}$.

We prove the lemma by induction on this tree from leaves to root.
Assume for each child $H_{i+1}^b$ of $H_i^a$, there is a feasible solution $S^b$ for the graph $M_{i+1}^b$ such that
$S^b = \left( \bigcup_{H \in T(H_{i+1}^b)} Sol(H) \right) \cup  \left( M_{i+1}^b \cap \left(\bigcup_{j \ge i+1} D_{f(j+1)} \right) \right)$.  
We prove that there is a feasible solution $S^a$ for $M_i^a$ such that $S^a = \left( \bigcup_{H \in T(H_{i}^a)} Sol(H) \right) \cup \left( M_i^a \cap \left( \bigcup_{j \ge i} D_{f(j+1)} \right) \right)$.
For the root $H_0^a$ of $T$, we have $M_0^a \cap \left(\bigcup_{j\ge 0} D_{f(j+1)} \right) \subseteq R$, and then the lemma follows from the case $i = 0$.

The base case is that $H_i^a$ is a leaf of $T$.
When $H_i^a$ is a leaf, there is no inner contracted node in $H_i^a$ and we have $M_i^a = H_i^a$. 
So $Sol(H_i^a)$ is a feasible solution for $M_i^a$.

We first need a claim the same as Claim~\ref{clm: EC1}.
Note that any inner contracted node of $H_i^a$ is enclosed by some cycle $C_b$.
Let $x$ be any inner contracted node of $H_i^a$ that is enclosed by $C_b$, and $X$ be the vertex set of the connected component of $G$ corresponding to $x$.
Then we have the following claim, whose proof is the same as that of Claim~\ref{clm: EC1}.
\begin{claim}\label{clm: VC1}
If $X \subseteq M_{i+1}^b$ for some $H_{i+1}^b$, then $(S^b \cup D_{f(i+1)}) \cap G[X]$ is connected.
\end{claim}

Now we ready to prove $S^a$ is a feasible solution for $M_i^a$.
That is, we prove it is triconnected.
Let $u$ and $v$ be any two vertices of $M_i^a$.
Let $Y_i^a = V(H_i^a) \setminus \{ \mbox{inner contracted nodes of } H_i^a \}$.
Since $V(M_i^a) = Y_i^a \cup \left( \bigcup_{H_{i+1}^b \mbox{ is a child of } H_i^a} V(M_{i+1}^b \setminus \{ r_{i+1}^b \})\right)$,
we have four cases.
\begin{description}
\item[Case 1: $\mathbf{u,v\in Y_i^a}$.] 
For any contracted component $X$ in $G$ that corresponds to an inner contracted node of $H_i^a$, by Claim~\ref{clm: VC1} all vertices in $X$ are connected in $(S^b \cup D_{f(i+1)}) \cap G[X]$ if $X \subseteq M_{i+1}^b$.
Then all vertices of $X$ are connected in $S^a \cap G[X]$, since for any child $H_{i+1}^b$ of $H_i^a$ we have $(S^b \cup D_{f(i+1)}) \subseteq S^a$.  
By the triconnectivity of $Sol(H_i^a)$, there are three vertex-disjoint paths between $u$ and $v$ in $Sol(H_i^a)$.
Since each inner contracted node of $H_i^a$ could be in only one path witnessing connectivity, the three vertex-disjoint $u$-to-$v$ paths in $Sol(H_i^a)$ could be transferred into another three vertex-disjoint $u$-to-$v$ paths in $S^a$ by replacing each contracted inner contracted node $x$ with a path in the corresponding component $X$. So $u$ and $v$ are triconnected in $S^a$.

\item[Case 2: $\mathbf{u,v \in M_{i+1}^b \setminus \{ r_{i+1}^b \}}$.] 
Since $V(C_b)$ is a cut for vertices enclosed by $C_b$ and those not enclosed by $C_b$, by the triconnectivity of $G$ we have $|V(C_b)| \ge 3$.
By inductive hypothesis, $S^b$ is a feasible solution for $M_{i+1}^b$, so there are three vertex-disjoint $u$-to-$r_{i+1}^b$ paths in $S^b$.
All these three paths must intersect $C_b$ by planarity, so there are three vertex-disjoint $u$-to-$C_b$ paths in $(S^b \setminus \{ r_{i+1}^b \}) \subseteq S^a$.
Similarly, there are three vertex-disjoint $v$-to-$C_b$ paths in $(S^b \setminus \{ r_{i+1}^b \} )\subseteq S^a$.
If we delete any two vertices in $S^a$, then there exist at least one $u$-to-$w_1$ path and one $v$-to-$w_2$ path for some vertices $w_1, w_2 \in C_b$.
Since all vertices in $C_b$ have level $f(i+1)$, they are in $Y_i^a$. Then by Case 1, vertices $w_1$ and $w_2$ are triconnected in $S^a$, so they are connected after deleting any two vertices.
Therefore, $u$ and $v$ are also connected after the deletion.

\item[Case 3: $\mathbf{u\in Y_i^a}$ and $\mathbf{v\in M_{i+1}^b \setminus \{ r_{i+1}^b\}}$.]
If one of $u$ and $v$ is in $C_b$, they are triconnected by Case 1 or 2.
So w.l.o.g. we assume $u$ is not enclosed by $C_b$ and $v$ is strictly enclosed by $C_b$.
Since $G$ is triconnected, we have $|V(C_b)| \ge 3$.
We could delete any two vertices in $S^a$ and there exists at least one vertex $w$ in $C_b$.
By Case 1, vertices $u$ and $w$ are connected after the deletion, and by Case 2, vertices $v$ and $w$ are connected after the deletion.
So $u$ and $v$ are connected after the deletion.

\item[Case 4: $\mathbf{u\in M_{i+1}^{b_1} \setminus \{ r_{i+1}^{b_1}\}}$ and $\mathbf{v\in M_{i+1}^{b_2} \setminus \{ r_{i+1}^{b_2}\}}$.]
W.l.o.g. assume $u$ is strictly enclosed by $C_{b_1}$ and $v$ is strictly enclosed by $C_{b_2}$,
otherwise, by Case 3 they are triconnected.
Since $G$ is triconnected, we have $|V(C_{b_1})| \ge 3$.
After deleting any two vertices in $S^a$, there exists a vertex $w \in C_{b_1}$.
By Case 2, vertices $u$ and $w$ are connected after deletion, and by Case 3 vertices $v$ and $w$ are connected after deletion.
So $u$ and $v$ are connected after deletion.
\end{description}
This completes the proof of Lemma~\ref{lem: VC2}.
\end{proof}
\fi

\begin{proof}[Proof of Theorem~\ref{thm: VC}]
We first prove the correctness, and then prove the running time.
Let the union $S = \left( \bigcup_{i\ge 0, C_a \in {\cal C}_i } \opt_w(H_i^a) \right) \cup  R$ be our solution.
By Lemma~\ref{lem: VC2}, the solution $S$ is feasible for $G$.
Then we have
$$\begin{array}{llll}
|S|
 & = \left| \left( \bigcup_{i\ge 0, C_a \in {\cal C}_i } \opt_w(H_i^a) \right) \cup  R \right| & \\ 
 & = \left| \left( \bigcup_{i\ge 0, C_a \in {\cal C}_i } \opt_w(H_i^a) \right) \right| + |R| & \\ 
 & \le \sum_{i \ge 0} \left| \bigcup_{ C_a \in {\cal C}_i } \opt_w(H_i^a)\right | +  | R | & \\ 
 & \le \sum_{i \ge 0} \left( \left| \opt(G) \cap (H_i \setminus R) \right| + |D_{f(i)}| + |D_{f(i+1)}| \right) + |R| &\mbox{by Lemma~\ref{lem: VC3}}\\
 & \le \sum_{i \ge 0} |\opt(G) \cap (H_i \setminus R) | + |R| + |R| + |R|  & \\ 
 & \le |\opt(G)| + 3|R| &\mbox{by~(\ref{equ: empty})}\\
 & \le (1+ 12/k) |\opt(G)|. &\mbox{by~(\ref{equ: VC1})}
\end{array}
$$
We set $k = 12/\epsilon$ and then we have $|S| \le (1+\epsilon) |\opt(G)|$.

Let $n = |V(G)|$ be the number of vertices in graph $G$.
We could find the edge set $R$ in linear time. 
By Lemma~\ref{lem: 3VC_slice} we could construct all slices in $O(n)$ time.
So the slicing step runs in linear time.
By Lemma~\ref{lem: bw}, the branchwidth of each slice is $O(k) = O(1/\epsilon)$.
Therefore, we could solve the minimum-weight 3-VCSS problem on each slice in linear time by Theorem~\ref{thm: dp}.
Based on the optimal solutions for all the slices, we could construct our final solution $S$ in linear time.
So our algorithm runs in linear time.
\end{proof}

\ifFull
\section{Dynamic Programming for Minimum-Weight $3$-ECSS on graphs with bounded branchwidth}\label{sec: dp}

In this section, we give a dynamic program to compute the optimal solution of minimum-weighted $3$-ECSS problem on a graph $G$ with bounded branchwidth $w$. 
This will prove Theorem~\ref{thm: dp} for the minimum-weight 3-ECSS problem.
Our algorithm is inspired by the work of Czumaj and Lingas~\cite{CL98, CL99}.
Note that $G$ need not be planar. 

Given a branch decomposition of $G$, we root its decomposition tree $T$ at an arbitrary leaf.  
For any edge $\alpha$ in $T$, let $L_{\alpha}$ be the separator corresponding to it, and $E_{\alpha}$ be the subset of $E(G)$ mapped to the leaves in the subtree of $T \setminus \{\alpha\}$ that does not include the root of $T$.
Let $H$ be a spanning subgraph of $G[E_{\alpha}]$. 
We adapt some definitions of Czumaj and Lingas~\cite{CL98, CL99}.
An {\em separator completion} of $\alpha$ is a multiset of edges between vertices of $L_{\alpha}$, each of which may appear up to $3$ times.  
\begin{definition}
A {\em configuration} of a vertex $v$ of $H$ for an edge $\alpha$ of $T$ is a pair $(A, B)$, where 
$A$ is a tuple $(a_1, a_2, \dots, a_{|L_{\alpha}|})$, representing that there are $a_i$ edge-disjoint paths from $v$ to the $i$th vertex of $L_{\alpha}$ in $H$, and $B$ is a set of tuples $(x_i, y_i, b_i)$, representing that there are $b_i$ edge-disjoint paths between the vertices $x_i$ and $y_i$ of $L_{\alpha}$ in $H$. (We only need those configurations where $|a_i| \le 3$ for all $0\le i \le |L_{\alpha}|$ and $|b_i| \le 3$ for all $i \ge 0$.) All the $\sum_{i=1}^{|L_{\alpha}|} a_i+\sum_{i}b_i$ paths in a configuration should be mutually edge-disjoint in $H$. 
\end{definition}
\begin{definition}
For any pair of vertices $u$ and $v$ in $H$, let $Com_H(u,v)$ be the set of separator completions of $\alpha$ each of which augments $H$ to a graph where $u$ and $v$ are three-edge connected. 
For each vertex $v$ in $H$, let $Path_H(v)$ be a set of configurations of $v$ for $\alpha$.
Let $Path_H$ be the set of all the non-empty $B$ in which all tuples can be satisfied in $H$.
Let $C_H$ be the set consisting of one value in each $Com_H(u,v)$ for all pairs of vertices $u$ and $v$ in $H$, and $P_H$ be the set consisting of one value in each $Path_H(v)$ for all vertices $v$ in $H$. 
We call the tuple $(C_H, P_H, Path_H)$ the {\em connectivity characteristic} of $H$, and denote it by $Char(H)$.
\end{definition}

Note that $|L_{\alpha}| \le w$ for any edge $\alpha$.
Subgraph $H$ may correspond to multiple $C_H$ and $P_H$, so $H$ may have multiple connectivity characteristics. Further, each value in $P_H$ represents at least one vertex.  
For any edge ${\alpha}$, there are at most $4^{O(w^2)}$ distinct separator completions ($O(w^2)$ pairs of vertices, each of which can be connected by at most $3$ parallel edges) and at most $2^{4^{O(w^2)}}$ distinct sets $C_H$ of separator completions. 
For any edge ${\alpha}$, there are at most $4^{O(w^2)}$ different configurations for any vertex in $H$ since the number of different sets $A$ is at most $4^w$, the number of different sets $B$ is at most $4^{O(w^2)}$ (the same as the number of separator completions). 
So there are at most $2^{4^{O(w^2)}}$ different sets of configurations $P_H$, and at most $2^{4^{O(w^2)}}$ different sets $B$. Therefore, there are at most $2^{4^{O(w^2)}}$ distinct connectivity characteristics for any edge $\alpha$. 

\begin{definition}
A configuration $(A, B)$ of vertex $v$ for $\alpha$ is {\em connecting} if the inequality $\sum_{i=1}^{|L|} a_i \geq 3$ holds where $a_i$ is the $i$th coordinate in $A$. That is, there are enough edge-disjoint paths from $v$ to the corresponding separator $L_{\alpha}$ which can connect $v$ and vertices outside $L_{\alpha}$. $Char(H)$ is connecting if all configurations in its $P_H$ set are connecting. Subgraph $H$ is connecting if at least one of $Char(H)$ is connecting. In the following, we only consider connecting subgraphs and their connecting connectivity characteristics.
\end{definition}

In the following, we need as a subroutine an algorithm to solve the following problem: when given a set of demands $(x_i, y_i, b_i)$ and a multigraph, we want to decide if there exist $b_i$ edge-disjoint paths between vertices $x_i$ and $y_i$ in the graph and all the $\sum_i b_i$ paths are mutually edge-disjoint. 
Although we do not have a polynomial time algorithm for this problem, we only need to solve this on graphs with $O(w)$ vertices, $O(w^2)$ edges and $O(w^2)$ demands.
So even an exponential time algorithm is acceptable for our purpose here. Let $ALG$ be an algorithm for this problem, whose running time is bounded by a function $f(w)$, which may be exponential in $w$.

For an edge $\alpha$ in the decomposition tree $T$, let $\beta$ and $\gamma$ be its two child edges. 
Let $H_{1}$ ($H_{2}$) be a spanning subgraph of $G[E_{\beta}]$ ($G[E_{\gamma}]$).  Let $H = H_1 \cup H_2$. Then we have the following lemma.
\begin{lemma}\label{lem: dp}
For any pair of $Char(H_1)$ and $Char(H_2)$, all the possible $Char(H)$, that could be obtained from $Char(H_1)$ and $Char(H_2)$, can be computed in 
$O(4^{w^2}f(w)+4^{w^24^{w^2}})$
time.
\end{lemma}
\begin{proof}
We compute all the possible sets for the three components of $Char(H)$.

\noindent{\bf Compute all possible $\mathbf{C_{H}}$} 
Each $C_{H}$ contains two parts: the first part covers all pairs of vertices in the same $H_i$ for $i=1,2$ and the second part covers all pairs of vertices from distinct subgraphs. 

For the first part, we generalize each value $C \in C_{H_i}$ for $i=1,2$ into a possible set $X_C$. Notice that each separator completion can be represented by a set of demands $(x, y, b)$ where $x$ and $y$ are in the separator. 
For a candidate separator completion $C'$ of ${\alpha}$, we combine $C'$ with each $B \in Path_{H_{3-i}}$ to construct a graph $H'$ and define the demand set the same as $C$. 
By running $ALG$ on this instance, we can check if $C'$ is a legal generalization for $C$. 
This could be computed in $4^{O(w^2)}w^2 + 4^{O(w^2)}f(w)$ time for each $C$.
All the legal generalizations for $C$ form $X_C$. 

Now we compute the second part. 
Let $(A_1, B_1) \in P_{H_1}$ and $(A_2, B_2) \in P_{H_2}$ be the configurations for some pair of vertices $u\in H_1$ and $v\in H_2$ respectively. We will compute possible $Com_H(u,v)$. 
We first construct a graph $H'$ on $L_{\beta} \cup L_{\gamma} \cup \{u, v\}$ by the two configurations: add $i$ parallel edges between two vertices if there are $i$ paths between them represented in the configurations. 
Then we check for each candidate separator completion $C'$ if $u$ and $v$ are three-edge connected in $H'\cup C'$.
We need $O(w^3)$ for this checking if we use Orlin's max-flow algorithm~\cite{Orlin13}.
All those $C'$ that are capable of providing three-edge-connectivity with $H'$ form $Com_H(u,v)$. This can be computed in $4^{O(w^2)}w^3$ time for each pair of configurations.

A possible $C_H$ consists of each value in $X_C$ for every $C \in C_{H_i}$ for $i =1,2$ and each value in $Com_H(u,v)$ for all pairs of configurations of $P_{H_1}$ and $P_{H_2}$. To compute all the sets, we need at most $4^{O(w^2)}w^3 + 4^{O(w^2)}f(w)$ time. There are at most $4^{O(w^2)}$ sets and each may contain at most $4^{O(w^2)}$ values. Therefore, to generate all the possible $C_H$ from those sets, we need at most $4^{w^24^{O(w^2)}}$ time.

\noindent{\bf Compute all possible $\mathbf{P_H}$}
We generalize each configuration $(A, B)$ of $v$ in $P_{H_i}$ ($i = 1,2$) into a set $Y_v$ of possible configurations. 
For each set $B'$ in $Path_{H_{3-i}}$, we construct a graph $H'$ by $A$, $B$ and $B'$ on vertex set $L_{\beta} \cup L_{\gamma} \cup \{v\}$: if there are $b$ disjoint paths between a pair of vertices represented in $A$, $B$ or $B'$, we add $b$ parallel edges between the same pair of vertices in $H'$, taking $O(w^2)$ time.
For a candidate value $(A^*, B^*)$ for ${\alpha}$, we define a set of demands according to $A^*$ and $B^*$ and run $ALG$ on all the possible $H'$ we construct for sets in $Path_{H_{3-i}}$. If there exists one such graph that satisfies all the demands, then we add this candidate value into $Y_v$. 
We can therefore compute each set $Y_v$ in $4^{O(w^2)}w^2 + 4^{O(w^2)}f(w)$ time.
A possible $P_H$ consists of each value in $Y_v$ for all $v\in V(H)$. There are at most $4^{O(w^2)}$ such sets and each may contain at most $4^{O(w^2)}$ values. So we can generate all possible $P_H$ from those sets in $4^{w^24^{O(w^2)}}$ time. 

\noindent{\bf Compute $\mathbf{Path_H}$}
For each pair of $B_1 \in Path_{H_1}$ and $B_2 \in Path_{H_2}$, we construct a graph $H'$ on vertex set $L_{\beta} \cup L_{\gamma}$: if two vertices are connected by $b$ disjoint paths, we add $b$ parallel edges between those vertices in $H'$. Since each candidate $B'$ for ${\alpha}$ can be represented by a set of demands, we only need to run $ALG$ on all possible $H'$ to check if $B'$ can be satisfied.
We add all satisfied candidates $B'$ into $Path_H$. This can be computed in $4^{O(w^2)}w^2 + 4^{O(w^2)}f(w)$ time.

Therefore, the total running time is $O(4^{w^2}f(w)+4^{w^24^{w^2}})$.
For each component we enumerate all possible cases, and the correctness follows.
\end{proof}

Our dynamic programming is guided by the decomposition tree $T$ from leaves to root. 
For each edge $\alpha$, our dynamic programming table is indexed by all the possible connectivity characteristics.
Each entry indexed by the connectivity characteristic $Char$ in the table is the weight of the minimum-weight spanning subgraph of $G[E_{\alpha}]$ that has $Char$ as its connectivity characteristic. 

\paragraph*{Base case} For each leaf edge $uv$ of $T$, the only subgraph $H$ is the edge $uv$ and the separator only contains the endpoints $u$ and $v$. 
$Com_H(u,v)$ contains the multisets of edge $uv$ that appears twice.
$Path_H(u)$ contains two configurations: $((3,0), \{(u, v, 1)\})$ and $((3,1), \emptyset )$, and $Path_H(v)$ contains two configurations: $((0, 3), \{(u, v, 1)\})$ and $((1,3), \emptyset)$. $Path_H$ contains one set: $\{(u, v, 1)\}$.

For each non-leaf edge $\alpha$ in $T$, we combine every pair of connectivity characteristics from its two child edges to fill in the dynamic programming table for $\alpha$. The root can be seen as a base case, and we can combine it with the computed results. 
The final result will be the entry indexed by $( \emptyset, \emptyset, \emptyset)$ in the table of the root.
Let $m = |E(G)|$. Then the size of the decomposition tree $T$ is $O(m)$. 
By Lemma~\ref{lem: dp}, we need $O(4^{w^2}f(w) + 4^{w^24^{w^2}})$ time to combine each pair of connectivity characteristics. 
Since there are at most $2^{4^{O(w^2)}}$ connectivity characteristics for each node, the total running time will be $O(2^{4^{w^2}}f(w)m + 4^{w^24^{w^2}}m)$. Since the branchwidth $w$ of $G$ is bounded, the running time will be $O(|E(G)|)$. 

\paragraph*{Correctness}
The separator completions guarantee the connectivity for the vertices in $H$, and the connecting configurations enumerate all the possible ways to connect vertices in $H$ and vertices of $V(G)\setminus V(H)$. So the connectivity requirement is satisfied. The correctness of the procedure follows from Lemma~\ref{lem: dp}. 

\fi

\subparagraph*{Acknowledgements}
We thank Glencora Borradaile and Hung Le for helpful discussions.
\newpage
\bibliography{zerr.bib}


\end{document}